\documentclass[a4paper,conference]{IEEEtran}

\IEEEsettopmargin{t}{30mm}
\IEEEquantizetextheight{c}
\IEEEsettextwidth{14mm}{14mm}
\IEEEsetsidemargin{c}{0mm}

\usepackage[utf8]{inputenc}
\usepackage[T1]{fontenc}
\usepackage{url}
\usepackage{ifthen}
\usepackage{cite}
\usepackage[cmex10]{amsmath} 
\usepackage{vkmacros}
\graphicspath{{figures/}}
\newif\ifshowtodo
\showtodotrue

\newif\ifshowdetail
\showdetailfalse

\newif\ifvlong
\vlongfalse

\newcommand{\VersionLength}{long}
\providecommand{\verlong}{\ifthenelse{\equal{\VersionLength}{long}}}
\newcommand{\VersionCols}{single}
\providecommand{\dcol}{\ifthenelse{\equal{\VersionCols}{double}}}

\def\noarrow{\ar@{-}[r]}

\title{Single-shot lossy compression: \\ mutual information bounds}

\author{%
  \IEEEauthorblockN{Victoria Kostina}
  \IEEEauthorblockA{California Institute of Technology\\
              Pasadena, CA, USA\\
              email: \href{mailto:vkostina@caltech.edu}{vkostina@caltech.edu}}
              \thanks{This work was supported by the Carver Mead New Adventure Grant. }
}

\IEEEoverridecommandlockouts

\date{}							
\begin{document}
\maketitle

\begin{abstract}
For several styles of fidelity constraints ---guaranteed distortion, conditional excess distortion, excess distortion--- we show mutual information upper bounds on the minimum expected description length needed to represent a random variable. Coupled with the corresponding converses, these results attest that as long as the information content in the data is not too low, minimizing the mutual information under an appropriate fidelity constraint serves as a reasonable proxy for the minimum description length of the data. We provide alternative characterizations of all three convex proxies,  shedding light on the structure of their solutions. 

\end{abstract}


\section{Introduction}

A variable-length \emph{quantizer} for the random variable $X \in \set X$ is a (usually non-injective, sometimes randomized) mapping $\mathsf f \colon \mathcal X \mapsto \set Y$, where $\set Y$ is called the reproduction alphabet. The fidelity of representation is quantified by a distortion measure, $\mathsf d \colon \set X \times \set Y \mapsto \mathbb R_+$, which, like the distribution of $X$, is given a priori. The quantizer is designed so as to satisfy the fidelity constraint (formulated as the maximum tolerated distortion in some probabilistic sense over the distribution of $X$). The efficiency of the quantizer can be quantified by the entropy of the discrete random variable at its output: the smaller the better as the entropy of a random variable is directly related to its minimum description length (see Appendix~\ref{apx:entropyvsmdl} for a precise statement of this classical wisdom). 

Minimizing the entropy over all quantizers satisfying the fidelity constraint means finding such a partition of $\set X$ that to each subset of the partition one can assign a representative from $\set Y$  such that the fidelity constraint is satisfied and the entropy is minimized over all such feasible partitions. It is a combinatorial optimization problem. In this paper we identify a convex proxy for the minimum entropy problem. The proxy optimization is a minimum mutual information problem. We show that the minimum quantizer entropy is sandwiched in terms of the proxy mutual information problem. For the proxy mutual information minimization, we provide a necessary and sufficient condition for the optimizer, thereby shedding light on the structure of its solution. 

We carry out the analysis for three styles of fidelity constraints. The first is guaranteed distortion, where a hard distortion threshold $d$ is imposed on the representation of $P_X$-a.e. realization $x$. The second is conditional excess distortion, where, for $P_X$-a.e. realization $x$, the quantizer is permitted to exceed distortion threshold $d$ with probability at most $\epsilon$, where the probability is taken over the quantizer’s randomness. The third is excess distortion, where the quantizer is permitted to exceed distortion threshold $d$ with probability at most $\epsilon$, but the probability is now taken jointly over both $X$ and the quantizer’s randomness; this last formulation is the most permissive out of the three, since violations are controlled only on average across source realizations.  

In terms of mutual information bounds on quantizer entropy, in \cite[Th. 2]{posner1971epsilon} such a bound is shown under guaranteed distortion; the bound includes an unspecified universal constant and requires $\mathsf d$ be a metric. Our result sharpens that bound, as our bound does not include unspecified constants and  does not impose requirements on $\mathsf d$. It is shown in \cite[Th. 2]{li2018strong} that if the fidelity constraint is expected distortion and encoder and decoder have access to common randomness, then the minimal quantizer entropy is bounded above by  $R_X(d) + \log_2 ( R_X(d) + 1) + 5 $ bits, where $R_X(d)$ is the minimal mutual information under expected distortion $d$. The form of this result is similar to ours, but we consider more stringent fidelity criteria and no common randomness.

Our characterizations of the proxy optimizations can be viewed as counterparts of Csisz\'ar's characterization of the minimal mutual information under expected distortion constraint \cite{csiszar1974extremum} to the three more stringent fidelity criteria mentioned above. 

The rest of the paper is organized as follows. \secref{sec:guaranteed} considers guaranteed distortion, \secref{sec:conditional} conditional excess distortion, and \secref{sec:excess} excess distortion.

Sets are denoted by calligraphic letters: $\set X$, $\set Y$; constants by lowercase letters: $x$, $y$; random variables by uppercase letters: $X$, $Y$; functions by sans font: $\mathsf f$, $\mathsf d$; the distribution of $X$ is $P_X$. Unless noted otherwise, $\log$ and $\exp$ are arbitrary common base. 


\section{Guaranteed distortion}
\label{sec:guaranteed}
The guaranteed distortion entropy  of the random variable $X$ is defined as \cite{posner1967epsilonentropy}: 
\begin{equation}
H_X(d, 0) \triangleq \inf_{\substack{ \mathsf f \colon \set X \mapsto \set Y \colon \\ \Prob{\mathsf d(X, \mathsf f (X) ) \leq d} = 1}} H(\mathsf f (X)).  \label{eq:Hd}
\end{equation}
Here $0$ in the left-hand side stands for the requirement that the distortion threshold is violated with probability~0. The quantity in \eqref{eq:Hd} is known as epsilon-entropy \cite{posner1967epsilonentropy}; it represents the minimum achievable entropy at the output of a quantizer that guarantees distortion $d$ for $P_X$-a.e. realization $x$. 

The minimal mutual information that we will show serves as convex proxy for \eqref{eq:Hd} is
\begin{align}
R_X(d, 0) &\triangleq \inf_{
 \substack{
 P_{Y|X} \colon \set X \mapsto \set Y\\
 \Prob{ \mathsf d(X, Y) \leq d} = 1
 }} I(X; Y) 
 \label{eq:objzero}
\end{align}
In \eqref{eq:objzero}, the deterministic quantizer mappings $\mathsf f$ are relaxed to transition probability kernels $P_{Y|X}$; and the entropy at the quantizer's output is replaced with the mutual information between input and output. Together, these relaxations convexify \eqref{eq:Hd}. Unsurprisingly, \eqref{eq:objzero} is a lower bound to \eqref{eq:Hd}; more surprisingly, \eqref{eq:Hd} is also upper bounded in terms of \eqref{eq:objzero}. 
The relevance of the minimization \eqref{eq:objzero} to approximating \eqref{eq:Hd} was noted already in \cite[Th. 2]{posner1971epsilon}, where it was shown that whenever $\mathsf d$ is a metric,
\begin{align}
 \!\! H_X(d, 0) \leq&~ R_X(d, 0) + \log_2 \mathbb ( R_X(d, 0) + 1) + C,  \label{Alossy0}
\end{align}
where $C$ is a universal constant. \thmref{thm:Hd}, stated next, sharpens Posner's result by refining the constant $C$ to $\log_2 e$; and we do not require $\mathsf d$ to be a metric.

\begin{thm}[guaranteed distortion: quantizer entropy]
\begin{align}
 {R}_{X} (d, 0) &\leq H_X(d, 0)  \label{eq:Hdlb} \\
 &\leq {R}_{X} (d, 0) + \log_2 \left({R}_{X} (d, 0) + 1 \right) 
\notag\\
&~+ \log_2 e   \quad \text{bits} \label{eq:Hdub}
\end{align}
\label{thm:Hd}
\end{thm}

\begin{proof}
The lower bound is noted in \cite[Th. 2]{posner1971epsilon}. To see it, let $Y = \mathsf f(X)$. Since $I(X; Y) \leq H(Y)$, the inequality \eqref{eq:Hdlb} follows by minimizing the left-hand side over $P_{Y|X}$ satisfying the constraint. 

For the upper bound, consider the function \cite{posner1971epsilon}
\begin{equation}
  {R}_{X}^+(d, 0) \triangleq \inf_{P_Y}\E{  -\log P_{Y}(\set B_d(X)) },  \label{RR(+)}
\end{equation}
where 
\begin{equation}
 \set B_d(x) \triangleq \left\{ y \in \set Y \colon \sd(x, y) \leq d \right\}
\end{equation}
is the distortion $d$-ball around $x$, and the infimum is over all distributions on $\set Y$. The ``$+$'' in the left side of \eqref{RR(+)} symbolizes that ${R}_{X}^+(d, 0)$ is represented as an expectation of a nonnegative random variable. 

To show the upper bound, we first establish that
\begin{align}
  {R}_{X}^+(d, 0) &\leq {R}_{X}(d, 0). \label{eq:Rdlb}
\end{align}
Indeed, for any $P_{XY}$ with $\Prob{\mathsf d(X, Y) \leq d} = 1$, by the data processing inequality of relative entropy, with the data processor that outputs $P_{Y|X = x}$ if $Y \in \set B_d(x)$ and $P_Y$ otherwise (occurs with probability 0 under $P_{Y|X = x}$), we have 
\begin{align}
D\left( P_{Y|X = x} \| P_{Y}\right) &\geq d(1 \| P_{Y}(\set B_d(x))) \label{eq:dpcondas}\\
&= \log \frac 1 {P_{Y}(\set B_d(x))}.
\end{align}
Here 
\begin{align}
d(\alpha \| q) &\triangleq \alpha \log \frac {\alpha} {q} + (1 - \alpha) \log \frac {1- \alpha} {1-q} 
\end{align}
is the binary relative entropy function. 

It is shown in \cite[(134)]{kostina2015varrate} (the below corrects the typo in \cite[(134)]{kostina2015varrate}, replacing 1 therein with $\log_2 e$)
\begin{align}
 H_{X}(d, 0)  
\leq&~
{R}_{X}^+(d, 0) + \log_2 \left({R}_{X}^+(d, 0) + 1 \right) 
\notag\\
+&~ \log_2 e  \quad \text{bits} \label{eq:Hdupper} 
\end{align}
Combining \eqref{eq:Hdupper} with \eqref{eq:Rdlb} leads to \eqref{eq:Hdub}. For completeness, we include the proof of \eqref{eq:Hdupper} in Appendix~\ref{apx:H<=R+}.
\end{proof}

The following result provides a characterization of the function $R_X(d, 0)$ and identifies a property of the optimal probability kernel that achieves it.

\begin{thm}[guaranteed distortion: minimal mutual information]
 \begin{equation}
 {R}_{X}(d, 0) = {R}_{X}^+(d, 0). \label{eq:RR+}
\end{equation}
Furthermore, the kernel $P_{Y|X}$ attains  ${R}_{X}(d, 0)$ if and only if for $P_X$-a.e. $x$
\begin{equation}
\log \frac{dP_{Y|X = x}}{dP_{Y}}(y) = 
\log   \frac{1}{P_{Y}(\set B_d(x))}   \label{eq:optkernel}
\end{equation}
where $P_Y$ is the $Y$ marginal of $P_X P_{Y|X}$. 
\label{thm:Rd}
\end{thm}
\begin{proof}
It is shown in \cite[Lemma 13]{posner1971epsilon} and in \cite[(132)]{kostina2015varrate} (proof included in Appendix~\ref{apx:RleqR+} for completeness)
 \begin{align}
 {R}_{X} (d, 0)
\leq&~ {R}_{X}^+(d, 0)   \label{eq:RleqR+},
\end{align}
Combining  \eqref{eq:Rdlb} with \eqref{eq:RleqR+} establishes \eqref{eq:RR+}. To show \eqref{eq:optkernel}, we note that equality in \eqref{eq:dpcondas} is achieved if and only if \eqref{eq:optkernel} holds. 
\end{proof}

\thmref{thm:Rd} parallels Csisz\'ar's characterization \cite{csiszar1974extremum} of the minimal information under expected distortion:
\begin{equation}
R_X(d) \triangleq \inf_{
 \substack{
 P_{Y|X} \colon \set X \mapsto \set Y\\
 \E{\mathsf d(X, Y)} \leq d
 }} I(X; Y)
 \label{eq:RXbar}
\end{equation}
Csisz{\'a}r's characterization states that 
\begin{equation}
R_X(d)   =  \inf_{P_Y} \max_{\lambda \geq 0} \E{ \Lambda_Y(X, \lambda) }
 \label{eq:csiszar}
\end{equation}
where 
\begin{equation}
\Lambda_Y(x, \lambda) \triangleq  \log \frac 1 {\E{ \exp (\lambda d - \lambda \sd(x, Y)}}
\end{equation}
is the generalized tilted information \cite[(28)]{kostina2011fixed}. The kernel $P_{Y|X}$ attains the infimum on the left-hand side of \eqref{eq:csiszar} if and only if
\begin{equation}
\log \frac{dP_{Y|X = x}}{dP_{Y}}(y) = \Lambda_Y(x, \lambda^\star) - \lambda^\star \sd(x, y) + \lambda^\star d
\end{equation}
 where $\lambda^\star = - R^\prime_X(d)$ is the negative of the derivative of $R_X(d)$ at $d$.
Csisz{\'a}r's characterization applies for all $d > d_{\min}$, where $d_{\min}$ is the infimum of $d$ values where  $R_X(d)$ is finite. 

Although the objective in \eqref{eq:objzero} can be viewed as a special case of \eqref{eq:RXbar} with distortion measure $\1{\mathsf \sd(x, y) > d}$ and distortion threshold $0$, \thmref{thm:Rd} is not a special case of Csisz{\'a}r's characterization because distortion threshold 0 is exactly the $d_{\min}$. 

The functions inside the expectations on the right-hand sides of \eqref{eq:optkernel} and \eqref{eq:csiszar} satisfy the relation (Markov's inequality)
\begin{equation}
\log   \frac{1}{P_{Y}(\set B_d(x))} \geq  \Lambda_Y(x, \lambda).
\end{equation}

\section{Conditional excess distortion}
\label{sec:conditional}
In this section, we weaken the constraint in \eqref{eq:Hd} by allowing the distortion threshold be violated with some probability, $\epsilon$, for (almost) every source realization $x$. Namely, we define the excess conditional distortion entropy of the source $X$ as: 
\begin{equation}
H_X^c(d, \epsilon) \triangleq \inf_{\substack{ P_{Y|X} \colon \set X \mapsto \set Y \colon \\ \Prob{\mathsf d(X, Y ) > d | X} \leq \epsilon  ~ \text{a.s.} } } H(Y).  \label{eq:Hepsdeltac}
\end{equation}
Here $c$ in the left side of \eqref{eq:Hepsdeltac} stands for conditioning the constraint on $X$. Note that in \eqref{eq:Hepsdeltac}, we allow randomization at the encoder. Without such randomization, the conditional probability in the constraint set of \eqref{eq:Hepsdeltac} would be either 0 or 1 depending on whether the deterministic representative for $X$ is within $d$ from $X$, meaning that for any $\epsilon < 1$, the feasible set of deterministic encoders in \eqref{eq:Hepsdeltac} coincides with that in \eqref{eq:Hd}.

This section establishes the following convex proxy for the minimization \eqref{eq:Hepsdeltac}:
\begin{equation}
R_X^c(d, \epsilon) \triangleq  \inf_{ \substack{ P_{Y | X } \colon \set X \mapsto \set Y \\ \Prob{ \mathsf d( X, Y) > d  | X} \leq \epsilon}} I(X; Y) \label{RR(eps)c}.
\end{equation} 

\begin{thm}[conditional excess distortion: quantizer entropy]
\begin{align}
 R_X^c (d, \epsilon) &\leq H_X^c(d, \epsilon)  \label{eq:Hdepslbc} \\
 &\leq R_X^c(d, \epsilon) + \log_2 \left(R_X^c(d, \epsilon) + 2 \right) 
\notag\\
&~+ 1 + \log_2 e   \quad \text{bits}  \label{eq:Hdepsubc}
\end{align}
\label{thm:Hdeps}
\end{thm}

\begin{proof}
The lower bound follows from $I(X; Y) \leq H(Y)$. 

To show the upper bound, we introduce the function
\begin{equation}
  {R}_{X}^{c\,+}(d,\epsilon) \triangleq \inf_{P_Y}\E{  d \left( \alpha_Y(X) \| P_Y\left(\set B_d(X)\right)\right) },\label{RR(+)c}
\end{equation}
where
\begin{equation}
\alpha_Y(x) \triangleq \max\left\{1 - \epsilon, P_Y(\set B_d(x)) \right\}.
\label{eq:alpha_Y}
\end{equation}
We establish that
\begin{align}
 {R}_{X}^{c\,+}(d,\epsilon)
  &\leq {R}_{X}^c(d, \epsilon). \label{eq:Rdepslbc}
\end{align}
Indeed, fix any $P_{Y| X = x}$ with $P_{Y| X = x}\left( \set B_d(x) \right)  \geq 1 - \epsilon$. With the data processor that outputs $P_{Y|X=x}$ if  $Y \in \set B_d(x)$ and $P_{Y}$ otherwise, data processing of relative entropy yields
\begin{align}
D\left( P_{Y|X = x} \| P_{Y}\right) &\geq d(P_{Y| X = x}\left( \set B_d(x) \right) \| P_Y(\set B_d(x))) \label{eq:dpca}\\
&\geq d( \alpha_Y(x) \| P_Y(\set B_d(x))), \label{eq:dpcb}
\end{align}
where to write \eqref{eq:dpcb} we used the fact that if $\alpha > q$, then $d(\alpha \| q)$ increases as a function of $\alpha$, and if  $\alpha < q$, then it decreases as $\alpha$ increases, vanishing to $0$ at the extreme $\alpha = q$.

Next, we assert the following extension of \eqref{eq:Hdupper}:
\begin{align}
H_X^c(d, \epsilon) 
\leq&~
 {R}_{X}^{c\,+}(d,\epsilon) + \log_2 \left(  {R}_{X}^{c\,+}(d,\epsilon) + 2 \right) 
\notag\\
+&~ 1 + \log_2 e \quad \text{bits}\label{eq:Hdepsupperc} 
\end{align}
Combining \eqref{eq:Hdepsupperc} with \eqref{eq:Rdepslbc} leads to \eqref{eq:Hdepsubc}. The proof of \eqref{eq:Hdepsupperc} is in Appendix~\ref{apx:H<=R+c}.
\end{proof}

\begin{thm}[conditional excess distortion: minimal mutual information]
 \begin{equation}
 R_X^c (d, \epsilon) = {R}_{X}^{c\,+}(d,\epsilon). \label{eq:RR+c}
\end{equation}
Furthermore, the kernel $P_{Y|X}$ attains  $R_X^c (d, \epsilon)$ if and only if for $P_X$-a.e. $x$
\begin{align}
\frac{dP_{Y|X = x}}{dP_{Y}}(y) &= 
\alpha_Y(x)  \frac{ \1{ \mathsf d(x, y) \leq d}  }{P_{Y}(\set B_d(x))}   \label{eq:optkernelc}\\
&~+ (1 - \alpha_Y(x)) \frac{\1{ \mathsf d(x, y) > d}}{P_{ Y}(\set B_d^c(x))}  \notag
\end{align}
where $P_Y$ is the $Y$ marginal of $P_X P_{Y|X}$. 
\label{thm:Rdc}
\end{thm}
\begin{proof}
It is shown in Appendix~\ref{apx:RleqR+c} that 
 \begin{align}
 {R}_{X}^c (d, \epsilon)
\leq&~ {R}_{X}^{c\,+}(d, \epsilon)   \label{eq:RleqR+c},
\end{align}
Combining  \eqref{eq:RleqR+c} with \eqref{eq:Rdepslbc} establishes \eqref{eq:RR+c}. To show \eqref{eq:optkernelc}, we note that equality in both \eqref{eq:dpca} and \eqref{eq:dpcb} is achieved if and only if \eqref{eq:optkernelc} holds. 
\end{proof}

\section{Excess  distortion}
\label{sec:excess}
In this section, we weaken the constraint in \eqref{eq:Hepsdeltac} by allowing the distortion threshold be violated with some probability, $\epsilon$, when the probability is averaged over all source realizations $X$:

\begin{equation}
H_X(d, \epsilon) \triangleq \inf_{\substack{ P_{Y|X} \colon \set X \mapsto \set Y \colon \\ \Prob{\mathsf d(X, Y) > d } \leq \epsilon } } H(Y).
\label{eq:Hepsdelta} 
\end{equation}
The function \eqref{eq:Hepsdelta} is similar to $(\epsilon,\delta)$-entropy \cite{posner1967epsilonentropy}, except \eqref{eq:Hepsdelta} allows randomized encoding mappings.

Consider the function  \cite[(19))]{kostina2015varrate}
\begin{equation}
R_X(d, \epsilon) \triangleq  \inf_{ \substack{ P_{Y | X } \colon \set X \mapsto \set Y \\ \Prob{ \mathsf d( X, Y) > d  } \leq \epsilon}} I(X; Y) \label{RR(eps)},
\end{equation} 

\begin{thm}[excess distortion: quantizer entropy]
\begin{align}
 {R}_{X} (d, \epsilon) &\leq H_X(d, \epsilon)  \label{eq:Hdepslb} \\
 &\leq {R}_{X}(d, \epsilon) + \log_2 \left({R}_{X}(d, \epsilon) + 2 \right) 
\notag\\
&~+ 1 + \log_2 e \quad \text{bits}  \label{eq:Hdepsub}
\end{align}
\label{thm:Hdepse}
\end{thm}

\begin{proof}
The lower bound is noted in \cite{kostina2015varrate} and follows from $I(X; Y) \leq H(Y)$. 

To show the upper bound, observe that all the steps in the proof of \thmref{thm:Hdeps} go through if the tolerated probability of violating distortion threshold is allowed to vary with source realization, i.e. if $\epsilon$ in \eqref{eq:Hepsdeltac}, \eqref{RR(eps)c} is replaced with $\epsilon(X)$. 

Next, notice that
\begin{align}
R_X(d, \epsilon) &=  \inf_{\substack{\epsilon \colon \set X \mapsto [0,1] \colon\\
\E{\epsilon(X)} \leq \epsilon}} R_X^c(d, \epsilon(\cdot)).
\label{eq:condvsexcess}
\end{align}
and analogously for $H_X(d, \epsilon)$. Since \eqref{eq:Hdepsubc} holds for an arbitrary set of thresholds $\epsilon(\cdot)$, it holds in particular for the optimal set of thresholds achieving the minimum on the right side of \eqref{eq:condvsexcess}; the validity of \eqref{eq:Hdepsub} follows.
\end{proof}

Define the function
\begin{equation}
  {R}_{X}^{+}(d,\epsilon) \triangleq \inf_{P_Y,\, \alpha_Y(\cdot)}
 \E{  d \left( \alpha_Y(X) \| P_Y\left(\set B_d(X)\right)\right) }.\label{RR(+)deps}
\end{equation}
where the infimum is over $P_Y$ and over functions $\alpha_Y \colon \mathcal X \mapsto [0,1]$ such that
\begin{align}
\alpha_Y(X) &\geq P_Y(\set B_d(X)) ~\text{a.s.} \label{eq:alpha1}\\
 \E{\alpha_Y(X)} &\geq 1 - \epsilon \label{eq:alpha2}
\end{align}

\begin{thm}[excess distortion: minimal mutual information]
 \begin{equation}
 R_X (d, \epsilon) = {R}_{X}^{+}(d,\epsilon). \label{eq:RR+e}
\end{equation}
Furthermore, the kernel $P_{Y|X}$ attains  ${R}_{X}(d, \epsilon)$ if and only if for $P_X$-a.e. $x$ \eqref{eq:optkernelc} holds, where
 $P_Y$ is the $Y$ marginal of $P_X P_{Y|X}$, and $\alpha_Y(\cdot)$ achieves the infimum on the right side of \eqref{RR(+)deps}.
\label{thm:Rde}
\end{thm}

\begin{proof}
The proof steps of \thmref{thm:Rdc} remain valid if, in \eqref{RR(eps)c} and \eqref{eq:alpha_Y}, the fixed threshold $\epsilon$ is replaced with variable threshold $\epsilon(X)$. Since equality \eqref{eq:RR+c} holds for an arbitrary set of thresholds, it also holds for the optimal set of thresholds. The constraint in \eqref{eq:alpha1} utilizes monotonicity properties of the binary relative entropy function mentioned after \eqref{eq:dpcb}.
\end{proof}

The optimum $\alpha_Y(\cdot)$ in \eqref{RR(+)deps} (and hence the optimum $\epsilon(\cdot)$ in \eqref{eq:condvsexcess}) can be analyzed explicitly: it is given by assigning $\alpha_Y(x) = 1$ to all typical $x$ and assigning $\alpha_Y(x)$ to its minimal value for the atypical $x$, where $x$ is considered typical if  $P_Y(\set B_d(x)) \geq q$, where threshold $q$ is the supremum of $q$'s such that
\begin{align}
\E{\alpha_Y(X)} &= \Prob{B \geq q} + \E{B\, \1{B < q}} \label{eq:Ealphaopt}\\
&\geq 1 - \epsilon
\end{align}
where $B \triangleq P_Y(\set B_d(X))$. For an (often reasonable) approximation to the optimum $\alpha_Y(\cdot)$, one may ignore the second term in the right side of \eqref{eq:Ealphaopt}; then $\log_2 \frac 1 q$ is the level-$(1-\epsilon)$ quantile of the distribution of ``ideal codelength'' $\log_2 \frac 1 B$. Indeed this is the path taken in \cite{kostina2015varrate}, where said quantile is used to define the ``$\epsilon$-cutoff'' of $\log_2 \frac 1 B$, which is shown to play a crucial rule in nonasymptotic tradeoffs in variable-rate lossy compression of i.i.d. sequences.

\bibliographystyle{IEEEtran}
\bibliography{../../../../Bibliography/references,../../../../Bibliography/vkjrnl}

\appendices
\section{Entropy vs. minimum expected length}
\label{apx:entropyvsmdl}
Denote by $ \left\{0, 1\right\}^\star$ the set of all binary strings (including the empty string). A lossless variable-length binary encoder for the discrete random variable $X \in \set X$ is an injective mapping $\mathsf f \colon \mathcal X \mapsto \left\{0, 1\right\}^\star$.  

Let 
\begin{equation}
 L^\star(X) \triangleq \min_{\mathsf f} \E{\mathrm{len}(\mathsf f(X))}
\end{equation}
be the expected length of the string output by the encoder, minimized over all lossless encoders $\mathsf f$. It is related to the entropy of $X$ via \cite{alon1994lower, wyner1972upper}
\begin{align}\label{eq:alon2}
		H(X) - \log_2 (H(X)+1) - \log_2 e &\le L^\star(X) \\
		&\le H(X). \label{eq:wyner}
\end{align}	

A lossless variable-length binary prefix-free encoder for the discrete random variable $X \in \set X$ is an injective mapping $\mathsf f \colon \mathcal X \mapsto \left\{0, 1\right\}^\star$, where the image of $\mathsf f$ (called the set of codewords) satisfies the condition that no codeword is a prefix of another codeword. 

Let $L^\star_p(X)$ denote the minimum expected encoded length achievable among all prefix-free encoders. It is related to the entropy $X$ via \cite{shannon1948mathematical}
\begin{align}
		H(X) &\le L^\star(X) \\
		&\le H(X) + 1.
\end{align}	

Thus, in either case ---with or without prefix constraints--- entropy of the random variable captures the storage capacity necessary for its lossless reproduction (by the decoder that knows the encoding function), as long as the entropy of the random variable is not too small, which is the case in most practical scenarios.

\section{Proof of \eqref{eq:Hdupper}}
\label{apx:H<=R+}
Let codewords $Y_1, Y_2, \ldots$ be drawn i.i.d. from $P_Y$. The encoder outputs a binary encoding of the first $d$-close match to $x$, i.e. 
\begin{equation}
W \triangleq \min \left\{ m \colon \mathsf d(x, Y_m) \leq d \right\} 
\end{equation}
If $y_1, y_2, \ldots$ is a realization of $Y^\infty$, $\mathsf f(x) = y_w$ is a deterministic mapping that satisfies the constraint in \eqref{eq:Hd}, so, since $w \mapsto y_w$ is injective, we have
\begin{align}
H_{d}(X) \leq H(W | Y^\infty = y^\infty) 
\end{align}
 
We proceed to show that $H(W | Y^\infty) $ is upper bounded by the right side of \eqref{eq:Hdupper}. Via the random coding argument this will imply that there exists at least one codebook $y^\infty$ such that $H(W | Y^\infty = y^\infty) $ is also upper bounded by the right side of \eqref{eq:Hdupper}, and the proof will be complete.

Let
\begin{align}
L \triangleq \lfloor \log_2 W \rfloor 
\end{align}
and consider the chain
\begin{align}
 H(W | Y^\infty) &\leq H(W) \label{eq:Hduppera1}\\
 &= H(W | L)  + I(W; L) \\
 &\leq \E{L} + H(L) \label{eq:Hdupperb1}\\
 &\leq \E{L} + \log_2 \left( 1 + \E{L}\right) + \log_2 e \label{eq:Hdupperc1}
\end{align}
where
\begin{itemize}
\item  \eqref{eq:Hduppera1} holds because conditioning decreases entropy; 
\item \eqref{eq:Hdupperb1} holds because conditioned on $L = \ell$, $W$ can have at most $2^\ell$ values; 
\item  \eqref{eq:Hdupperc1} holds because the entropy of a positive integer-valued random variable with a given mean is maximized by the geometric distribution. 
\end{itemize}

Finally,
\begin{align}
\E{L} &\leq  \E{ - \log_2 P_{Y}(\set B_d(X)) }
\label{eq:ELub}
\end{align}
This is because 
\begin{align}
 \E{L | X = x} \leq&~ \log_2 \E{W | X = x}  \\
 =&~  - \log_2 P_{Y}(\set B_d(x)), \label{eq:ELubc}
\end{align}
which applies Jensen's inequality and the fact that conditioned on $X = x$ and averaged over codebooks, $W$ has geometric distribution with success probability $P_Y(\set B_d(x))$.

\section{Proof of~\eqref{eq:RleqR+}}
\label{apx:RleqR+}
 Fix a distribution $P_{\bar Y}$ on $\set Y$ and define the conditional probability distribution $P_{Y|X}$ through\footnote{Note that in general $P_X \to P_{Y|X} \nrightarrow P_{\bar Y}$.}
\begin{equation}
\frac{dP_{Y|X = x}}{dP_{\bar Y}}(y) = 
  \frac{\1{ \mathsf d(x, y) \leq d}}{P_{ \bar Y}(\set B_d(x))}   \label{PZ|Slossy}.
\end{equation}

Upper-bounding the minimum in \eqref{eq:objzero} with the choice of $P_{Y|X}$ in \eqref{PZ|Slossy}, we obtain 
\begin{align} 
{R}_{X}(d, 0) 
\leq&~ 
I(X; Y) 
\\
=&~ 
D\left( P_{Y | X} \| P_{ \bar Y} | P_{X} \right) - D(P_{Y} \| P_{ \bar Y})\\
\leq&~ D\left( P_{Y | X} \| P_{\bar Y } | P_{X} \right) \\
 =&~ \E{  - \log P_{\bar Y}(\set B_d(X)) } 
\end{align}
which leads to \eqref{eq:RleqR+} after minimizing the right side over all $P_{\bar Y}$.

\section{Proof of \eqref{eq:Hdepsupperc}}
\label{apx:H<=R+c}
Consider an encoder that, given an infinite list of codewords $y_1, y_2, \ldots$, outputs the first $d$-close match to $x$ with probability $\alpha_Y(x)$, and outputs $y_1$ otherwise. Specifically,  the encoder outputs a binary encoding of 
\begin{equation}
W \triangleq 
\begin{cases}
\min \left\{ m \colon \mathsf d(x, y_m) \leq d \right\} &   \text{w.p.~} \alpha_Y(x) \\
1 & \text{otherwise}
\end{cases}
\label{eq:Weps}
\end{equation}
This rather trivial randomized encoder, which simply gives up with a pre-determined probability and produces a $d$-close match otherwise, is indeed feasible in the context of the optimization problem \eqref{eq:Hepsdeltac}. 

The reasoning in Appendix~\ref{apx:H<=R+} applies to bound the entropy of $W$ in \eqref{eq:Weps}, with \eqref{eq:ELubc} replaced by
\begin{align}
 \E{L | X = x} &\leq d( \alpha_Y(x) \| P_Y(\set B_d(x))) + h(\alpha_Y(x))
 \label{eq:ELubcc}
\end{align}
and applying $h(\alpha_Y(x)) \leq 1$ bit, a crude bound resulting in the extra $+1$'s in \eqref{eq:Hdepsupperc} compared to \eqref{eq:Hdupper}, as in typical scenarios --- a small $\epsilon$, an even smaller $P_Y(\set B_d(x))$ --- $h(\alpha_Y(x))$ would be nearly 0. 

To show \eqref{eq:ELubcc}, denote by $U$ the Bernoulli random variable with success probability $\alpha_Y(X)$ employed in \eqref{eq:Weps}.  By the law of iterated expectation and Jensen's inequality, 
\begin{align}
\E{L | X = x} &= \E{\E{L | X = x, U}} \\
&\leq \E{ \log_2 \E{W | X = x, U}} \\
&=\alpha_Y(x) \log_2 \frac 1 {P_Y(\set B_d(x)))},
\label{eq:Ucond}
\end{align} 
where \eqref{eq:Ucond} uses that averaged over the codebook and conditioned on $X = x, U = 1$, the waiting time $W$ has geometric distribution with success probability $P_Y(\set B_d(x))$; while conditioned on $X = x, U = 0$, it is equal to $0$. 
The right side of \eqref{eq:ELubcc} is an upper bound to the right side of \eqref{eq:Ucond}.

\section{Proof of \eqref{eq:RleqR+c}}
\label{apx:RleqR+c}
 Fix a distribution $P_{\bar Y}$ on $\set Y$ and define the conditional probability distribution $P_{Y|X}$ through
\begin{align}
\frac{dP_{Y|X = x}}{dP_{\bar Y}}(y) =&~ 
\alpha_{\bar Y}(x)  \frac{ \1{ \mathsf d(x, y) \leq d}  }{P_{ \bar Y}(\set B_d(x))}  \label{PZ|Slossyc} \\
 &~+ (1 - \alpha_{\bar Y}(x)) \frac{\1{ \mathsf d(x, y) > d}}{P_{ \bar Y}(\set B_d^c(x))} \notag.
\end{align}

Upper-bounding the minimum in \eqref{RR(eps)c} with the choice of $P_{Y|X}$ in \eqref{PZ|Slossyc}, we obtain 
\begin{align}
{R}_{X}^c(d, \epsilon) 
\leq&~ 
I(X; Y) 
\\
\leq&~ D\left( P_{Y | X} \| P_{\bar Y } | P_{X} \right) \\
=&~ \E{  d \left( \alpha_Y(X) \| P_{\bar Y}\left(\set B_d(X)\right)\right) } 
\end{align}
which leads to \eqref{eq:RleqR+c} after minimizing the right side over all $P_{\bar Y}$.

\end{document}